\documentclass[conference,letter,romanappendices]{ieeeconf}
\let\proof\relax   

\usepackage{amsthm,xpatch}
\usepackage{amsmath,amsfonts}
\usepackage{cite}
\usepackage{amssymb}
\usepackage{dsfont}
\usepackage{graphicx, subfigure}
\usepackage{color}
\usepackage{breqn}
\usepackage{mathtools}
\usepackage{bbm}
\usepackage{latexsym}
\usepackage[ruled, linesnumbered]{algorithm2e}
\usepackage{accents}
\usepackage{tikz}

\newtheorem{lemma}{Lemma}
\newtheorem{theorem}{Theorem}
\newtheorem{remark}{Remark}

\newcommand*{\transpose}{%
  {\mathpalette\@transpose{}}%
}

\IEEEoverridecommandlockouts

\begin{document}

\newcommand{\SB}[3]{
\sum_{#2 \in #1}\biggl|\overline{X}_{#2}\biggr| #3
\biggl|\bigcap_{#2 \notin #1}\overline{X}_{#2}\biggr|
}

\newcommand{\Mod}[1]{\ (\textup{mod}\ #1)}

\newcommand{\overbar}[1]{\mkern 0mu\overline{\mkern-0mu#1\mkern-8.5mu}\mkern 6mu}

\makeatletter
\newcommand*\nss[3]{%
  \begingroup
  \setbox0\hbox{$\m@th\scriptstyle\cramped{#2}$}%
  \setbox2\hbox{$\m@th\scriptstyle#3$}%
  \dimen@=\fontdimen8\textfont3
  \multiply\dimen@ by 4             
  \advance \dimen@ by \ht0
  \advance \dimen@ by -\fontdimen17\textfont2
  \@tempdima=\fontdimen5\textfont2  
  \multiply\@tempdima by 4
  \divide  \@tempdima by 5          
  \ifdim\dimen@<\@tempdima
    \ht0=0pt                        
    \@tempdima=\fontdimen5\textfont2
    \divide\@tempdima by 4          
    \advance \dimen@ by -\@tempdima 
    \ifdim\dimen@>0pt
      \@tempdima=\dp2
      \advance\@tempdima by \dimen@
      \dp2=\@tempdima
    \fi
  \fi
  #1_{\box0}^{\box2}%
  \endgroup
  }
\makeatother

\makeatletter
\renewenvironment{proof}[1][\proofname]{\par
  \pushQED{\qed}%
  \normalfont \topsep6\p@\@plus6\p@\relax
  \trivlist
  \item[\hskip\labelsep
        \itshape
    #1\@addpunct{:}]\ignorespaces
}{%
  \popQED\endtrivlist\@endpefalse
}
\makeatother

\makeatletter
\newsavebox\myboxA
\newsavebox\myboxB
\newlength\mylenA

\newcommand*\xoverline[2][0.75]{%
    \sbox{\myboxA}{$\m@th#2$}%
    \setbox\myboxB\null
    \ht\myboxB=\ht\myboxA%
    \dp\myboxB=\dp\myboxA%
    \wd\myboxB=#1\wd\myboxA
    \sbox\myboxB{$\m@th\overline{\copy\myboxB}$}
    \setlength\mylenA{\the\wd\myboxA}
    \addtolength\mylenA{-\the\wd\myboxB}%
    \ifdim\wd\myboxB<\wd\myboxA%
       \rlap{\hskip 0.5\mylenA\usebox\myboxB}{\usebox\myboxA}%
    \else
        \hskip -0.5\mylenA\rlap{\usebox\myboxA}{\hskip 0.5\mylenA\usebox\myboxB}%
    \fi}
\makeatother

\xpatchcmd{\proof}{\hskip\labelsep}{\hskip3.75\labelsep}{}{}

\pagestyle{plain}

\title{\fontsize{21}{28}\selectfont Capacity of Single-Server Single-Message Private Information Retrieval with Private Coded Side Information}

\author{Anoosheh Heidarzadeh, Fatemeh Kazemi, and Alex Sprintson\thanks{The authors are with the Department of Electrical and Computer Engineering, Texas A\&M University, College Station, TX 77843 USA (E-mail: \{anoosheh, fatemeh.kazemi, spalex\}@tamu.edu).}}

%


\maketitle 

\thispagestyle{plain}

\begin{abstract} 
We study the problem of single-server single-message Private Information Retrieval with Private Coded Side Information (PIR-PCSI). In this problem, there is a server that stores a database, and a user who knows a random linear combination of a random subset of messages in the database. The number of messages contributing to the user's side information is known to the server a priori, whereas their indices and coefficients are unknown to the server a priori. The user wants to retrieve a message from the server (with minimum download cost), while protecting the identities of both the demand and side information messages. 

Depending on whether the demand is part of the coded side information or not, we consider two different models for the problem. For the model in which the demand does not contribute to the side information, we prove a lower bound on the minimum download cost for all (linear and non-linear) PIR protocols; and for the other model wherein the demand is one of the messages contributing to the side information, we prove a lower bound for all scalar-linear PIR protocols. In addition, we propose novel PIR protocols that achieve these lower bounds.
\end{abstract}

\section{introduction}
In the information-theoretic Private Information Retrieval (PIR) problem (see, e.g.,~\cite{Sun2017,JafarPIR3new}), there is a user that wishes to download a single or multiple messages belonging to a database stored on a single or multiple (non-colluding or colluding) servers. The goal of the user is to minimize the download cost (i.e., the amount of information downloaded from the server(s)), while hiding the identity of its demanded message(s) from the server(s). This setup was recently extended in~\cite{Kadhe2017,Tandon2017,Wei2017CacheAidedPI,Wei2017FundamentalLO,KGHERS2017,Chen2017side,HKS2018,Maddah2018,HKGRS:2018,LG:2018} to the settings wherein the user has some side information about the messages in the database, and the side information is unknown to the server(s). 

For the single-server setting of the PIR problem in the presence of some side information, we studied the cases in which the side information is a random subset of messages (a.k.a. PIR with Side Information (PIR-SI)) or a random linear combination of a random subset of messages (a.k.a. PIR with Coded Side Information (PIR-CSI)) in~\cite{Kadhe2017,HKGRS:2018} and \cite{HKS2018}, respectively. The multi-server setting of the PIR-SI problem was also studied in~\cite{KGHERS2017,Chen2017side,Maddah2018}. For the PIR-SI problem, two different types of privacy, known as \emph{$W$-privacy} (i.e., only the identities of the demand messages must be protected) and \emph{$(W,S)$-privacy} (i.e., the identities of both the demand and side information messages must be protected jointly) have been considered, whereas the problem of PIR-CSI has only been studied when $W$-privacy is required. 

In this work, we study the single-server single-message PIR-CSI problem where $(W,S)$-privacy is required. In this problem, referred to as \emph{PIR with Private Coded Side Information (PIR-PCSI)}, there is a single server storing a database of $K$ messages, and there is a user who knows a random linear combination of a random subset of $M$ messages. This setting can be motivated by several practical scenarios. The user may have obtained their side information via overhearing in a wireless network; or from a trusted server with limited knowledge about the database; or from the information locally stored in the user's cache of limited size, to name a few. The user is interested in downloading a single message from the server while preserving the privacy of both the demand message and the messages contributing to the side information. Depending on whether the user's demanded message itself contributes to the user's side information or not, we consider two different models of the PIR-PCSI problem. 

\subsection{Main Contributions}
For the model in which the demanded message is not part of the coded side information, we characterize the capacity and the scalar-linear capacity of the PIR-PCSI problem, where the (scalar-linear) capacity is defined as the supremum of all achievable rates (i.e., the inverse of the download cost) for all (scalar-linear) protocols. In particular, we show that for this model the capacity and the scalar-linear capacity are both equal to $(K-M)^{-1}$ for any ${0\leq M\leq K-1}$. This is interesting because, as shown in~\cite[Theorem~2]{Kadhe2017}, even when the user knows $M$ (uncoded) messages as their side information, in order to guarantee $(W,S)$-privacy, the minimum download cost is $K-M$. This shows that for achieving $(W,S)$-privacy there will be no loss in capacity even if only \emph{one} linear combination of $M$ messages (instead of $M$ messages separately) is known to the user a priori. 

For the model wherein the user's demanded message contributes to their coded side information, we show that the scalar-linear capacity of the PIR-PCSI problem is equal to $(K-M+1)^{-1}$ for any $2\leq M\leq K$. Interestingly, this result shows that when the user knows $M-1$ messages (different from the demand), achieving $(W,S)$-privacy is as costly as that when the user knows only \emph{one} linear combination of the $M-1$ messages and the demand. 

The converse proofs are based on information-theoretic arguments, and the proofs of achievability rely on novel PIR protocols based on the Generalized Reed-Solomon (GRS) codes that include a specific codeword.

\section{Problem Formulation}\label{sec:SN}
Let $\mathbb{F}_q$ be a finite field of size $q$, and let $\mathbb{F}_{q^m}$ be an extension field of $\mathbb{F}_q$ for some integer $m$. Let $L \triangleq m\log_2 q$, and let $\mathbb{F}_q^{\times} \triangleq \mathbb{F}_q\setminus \{0\}$. For a positive integer $i$, we denote $\{1,\dots,i\}$ by $[i]$. Let $K\geq 1$ and $0\leq M\leq K$ be two integers. We denote the set of all subsets of $\mathcal{K}\triangleq [K]$ of size $M$ by $\mathcal{S}$, and the set of all sequences of length $M$ with elements from $\mathbb{F}^{\times}_q$ by $\mathcal{C}$. 

Assume that there is a server that stores a set of $K$ messages $X_1,\dots,X_K$, with each message $X_i$ being independently and uniformly distributed over $\mathbb{F}_{q^m}$, i.e., ${H(X_1) = \dots = H(X_K) = L}$ and $H(X_1,\dots,X_K) = KL$. Also assume that there is a user that wishes to retrieve a message $X_W$ from the server for some $W\in \mathcal{K}$, and knows a linear combination ${Y^{[S,C]}\triangleq \sum_{i\in S} c_i X_i}$ for some $S \triangleq \{i_1,\dots,i_M\}\in \mathcal{S}$ and ${C \triangleq \{c_{i_1},\dots,c_{i_M}\} \in \mathcal{C}}$. We refer to $W$ as the \emph{demand index}, $X_W$ as the \emph{demand}, $S$ as the \emph{side information index set}, $Y^{[S,C]}$ as the \emph{side information}, and $M$ as the \emph{side information size}. 

We denote by $\boldsymbol{S}$, $\boldsymbol{C}$, and $\boldsymbol{W}$ the random variables representing $S$, $C$, and $W$, respectively. We also denote the probability mass function (PMF) of $\boldsymbol{S}$ by $p_{\boldsymbol{S}}(\cdot)$, the PMF of $\boldsymbol{C}$ by $p_{\boldsymbol{C}}(\cdot)$, and the conditional PMF of $\boldsymbol{W}$ given $\boldsymbol{S}$ by $p_{\boldsymbol{W}|\boldsymbol{S}}(\cdot|\cdot)$. We assume that $\boldsymbol{S}$ is uniformly distributed over $\mathcal{S}$, i.e., $p_{\boldsymbol{S}}(S) = \binom{K}{M}^{-1}$ for all $S\in \mathcal{S}$; and $\boldsymbol{C}$ is uniformly distributed over $\mathcal{C}$, i.e., $p_{\boldsymbol{C}}(C) = (q-1)^{-M}$ for all $C\in\mathcal{C}$. Also, we consider two different models for the conditional PMF of $\boldsymbol{W}$ given $\boldsymbol{S}=S$ as follows: 

\subsubsection*{Model~I} $\boldsymbol{W}$ is uniformly distributed over $\mathcal{K}\setminus S$, i.e., 
\begin{equation*}
p_{\boldsymbol{W}|\boldsymbol{S}}(W|S) = 
\left\{\begin{array}{ll}
(K-M)^{-1}, & W\not\in S,\\	
0, & \text{otherwise}.
\end{array}\right.	
\end{equation*} 

\subsubsection*{Model~II} $\boldsymbol{W}$ is uniformly distributed over $S$, i.e., 
\begin{equation*}
p_{\boldsymbol{W}|\boldsymbol{S}}(W|S) = 
\left\{\begin{array}{ll}
M^{-1}, & W\in S,\\	
0, & \text{otherwise};
\end{array}\right.	
\end{equation*} To avoid the degenerate cases, we assume ${0\leq M\leq K-1}$ and $2\leq M\leq K$ for the models~I and~II, respectively.

Let $I^{[W,S]}$ be an indicator function such that ${I^{[W,S]} = 1}$ if $W\in S$, and ${I^{[W,S]} = 0}$ if $W\not\in S$. Note that $\mathbb{P}(\boldsymbol{W}=W',\boldsymbol{S}=S'|I^{[W,S]}=0)$ is equal to ${(K-M)^{-1}\binom{K}{M}^{-1}}$ if $W'\not\in S'$, and it is zero otherwise; and $\mathbb{P}(\boldsymbol{W}=W',\boldsymbol{S}=S'|I^{[W,S]}=1)$ is equal to ${M^{-1}\binom{K}{M}^{-1}}$ if $W'\in S'$, and it is zero otherwise.  

We assume that $I^{[W,S]}$ is known to the server a priori. We also assume that the server knows the size of $S$ (i.e., $M$) and the PMF's $p_{\boldsymbol{S}}(\cdot)$, $p_{\boldsymbol{C}}(\cdot)$, and $p_{\boldsymbol{W}|\boldsymbol{S}}(\cdot|\cdot)$, whereas the realizations $S$, $C$, and $W$ are unknown to the server a priori. 

For any $S$, $C$, and $W$, in order to retrieve $X_W$, the user sends to the server a query $Q^{[W,S,C]}$, which is a (potentially stochastic) function of $W$, $S$, $C$, and $Y^{[S,C]}$. The query $Q^{[W,S,C]}$ must protect the privacy of both the user's demand index $W$ and side information index set $S$ from the server's perspective, i.e., for any given $\theta\in \{0,1\}$,
\begin{align*}& \mathbb{P}(\boldsymbol{W}= W',\boldsymbol{S}=S'| Q^{[W,S,C]},I^{[W,S]}=\theta) \\ & \quad = \mathbb{P}(\boldsymbol{W}= W',\boldsymbol{S}=S'| I^{[W,S]}=\theta)\end{align*} for all $W'\in \mathcal{K}$ and all $S'\in \mathcal{S}$. We refer to this condition as the \emph{$(W,S)$-privacy condition}. Note that the $(W,S)$-privacy condition is stronger than the $W$-privacy condition being previously studied in~\cite{HKS2018}, where the query must protect only the privacy of the user's demand index, i.e., for any given $\theta\in \{0,1\}$, we have $\mathbb{P}(\boldsymbol{W}=W'|Q^{[W,S,C]},I^{[W,S]}=\theta)=\mathbb{P}(\boldsymbol{W}=W'|I^{[W,S]}=\theta)$ for all $W'\in \mathcal{K}$ and all $S'\in \mathcal{S}$. 

Upon receiving $Q^{[W,S,C]}$, the server sends to the user an answer $A^{[W,S,C]}$, which is a (deterministic) function of the query $Q^{[W,S,C]}$, the indicator $I^{[W,S]}$, and the messages in $X$, i.e., $H(A^{[W,S,C]}| Q^{[W,S,C]},I^{[W,S]}, \{X_i\}_{i\in \mathcal{K}}) = 0$. The answer $A^{[W,S,C]}$ along with the query $Q^{[W,S,C]}$, the indicator $I^{[W,S]}$, and the side information $Y^{[S,C]}$ must enable the user to retrieve the demand $X_W$, \[H(X_W| A^{[W,S,C]}, Q^{[W,S,C]}, I^{[W,S]}, Y^{[S,C]})=0.\] This condition is referred to as the \emph{recoverability condition}. 

For each model (I or~II), the problem is to design a query $Q^{[W,S,C]}$ and an answer $A^{[W,S,C]}$ for any $W$, $S$, and $C$ that satisfy the privacy and recoverability conditions. We refer to this problem as \emph{single-server single-message Private Information Retrieval (PIR) with Private Coded Side Information (PCSI)}, or \emph{PIR-PCSI} for short. Specifically, we refer to the PIR-PCSI problem under the model~I as \emph{PIR-PCSI--I}, and under the model~II as \emph{PIR-PCSI--II}.

We refer to a collection of $Q^{[W,S,C]}$ and $A^{[W,S,C]}$ (for all $W$, $S$, and $C$ such that $I^{[W,S]}=0$ or $I^{[W,S]}=1$) which satisfy the privacy and recoverability conditions as a \emph{PIR-PCSI--I protocol} or a \emph{PIR-PCSI--II protocol}, respectively. 

The \emph{rate} of a PIR-PCSI (--I or --II) protocol is defined as the ratio of the entropy of a message, i.e., $L$, to the average entropy of the answer, i.e., $H(A^{[\boldsymbol{W},\boldsymbol{S},\boldsymbol{C}]})=\sum H(A^{[W,S,C]})p_{\boldsymbol{W}|\boldsymbol{S}}(W|S)p_{\boldsymbol{S}}(S)p_{\boldsymbol{C}}(C)$, where the summation is over all $W$, $S$, and $C$ (such that $I^{[W,S]} = 0$ or $I^{[W,S]} = 1$). The \emph{capacity} of PIR-PCSI (--I or --II) problem is defined as the supremum of rates over all PIR-PCSI (--I or --II) protocols. The supremum of rates over all scalar-linear PIR-PCSI (--I or --II) protocols, i.e., the answer contains only scalar-linear combinations of the messages, is defined as the \emph{scalar-linear capacity} of PIR-PCSI (--I or --II) problem. 

In this work, our goal is to characterize the capacity and the scalar-linear capacity of the PIR-PCSI--I and PIR-PCSI--II problems, and to design PIR-PCSI (--I and --II) protocols that are capacity-achieving. 

\section{Main Results}
We present our main results in this section. The capacity and the scalar-linear capacity of PIR-CSI--I problem are characterized in Theorem~\ref{thm:PIRPCSI-I}, and the scalar-linear capacity of PIR-CSI--II problem is characterized in Theorem~\ref{thm:PIRPCSI-II}. The proofs are given in Sections~\ref{sec:PIRPCSI-I} and~\ref{sec:PIRPCSI-II}.

\begin{theorem}\label{thm:PIRPCSI-I}
The capacity and the scalar-linear capacity of PIR-PCSI--I problem with $K$ messages and side information size $0\leq M\leq K-1$ are given by $(K-M)^{-1}$.
\end{theorem}

The converse follows directly from the result of~\cite[Theorem~2]{Kadhe2017}, which was proven using an index coding argument, for single-server single-message PIR with (uncoded) side information when $(W,S)$-privacy is required. In this work, we provide an alternative proof by upper bounding the rate of any PIR-PCSI--I protocol using information-theoretic arguments (see Section~\ref{subsec:ConvThm1}). The key component of the proof is a necessary condition implied by the $(W,S)$-privacy and recoverability conditions (see Lemma~\ref{prop:1}). The achievability proof relies on a new PIR-PCSI--I protocol, termed the \emph{Specialized GRS Code protocol}, based on the Generalized Reed-Solomon (GRS) codes with a specific codeword, which achieves the rate $(K-M)^{-1}$ (see Section~\ref{subsec:AchThm1}).  

\begin{remark}
\emph{It was shown in~\cite{Kadhe2017} that when there is a single server storing $K$ messages, and there is a user that knows $M$ (uncoded) messages as their side information and demands a single message not in their side information, in order to guarantee the $(W,S)$-privacy condition, the minimum download cost is $K-M$. Surprisingly, this result matches the result of Theorem~\ref{thm:PIRPCSI-I}. This shows that for achieving $(W,S)$-privacy there will be no loss in capacity even if only \emph{one} linear combination of $M$ messages (instead of $M$ messages separately) is known to the user a priori.}
\end{remark}

\begin{remark}
\emph{When $W$-privacy, which is a weaker notion of privacy in comparison to $(W,S)$-privacy, is required (i.e., only the user's demand index, and not the user's side information index set, must be protected from the server), the result of~\cite[Theorem~1]{HKS2018} shows that the capacity of single-server single-message PIR with a coded side information that does not include the demand (known as the PIR-CSI--I problem in~\cite{HKS2018}) is equal to $\lceil\frac{K}{M+1}\rceil^{-1}$. Since $\lceil\frac{K}{M+1}\rceil< K-M$ for all ${1\leq M\leq K-2}$, the capacity of PIR-PCSI--I is strictly smaller than that of PIR-CSI--I, as expected. However, for the two extremal cases of $M=0$ and $M=K-1$, it follows that $(W,S)$-privacy comes at no extra cost than $W$-privacy.}
\end{remark}

\begin{theorem}\label{thm:PIRPCSI-II}
The scalar-linear capacity of PIR-PCSI--II problem with $K$ messages and side information size ${2\leq M\leq K}$ is given by $(K-M+1)^{-1}$.
\end{theorem}

The converse proof is based on a mixture of algebraic and information-theoretic arguments (see Section~\ref{subsec:ConvThm2}), and the proof of achievability is based on a modified version of the Specialized GRS Code protocol which achieves the rate $(K-M+1)^{-1}$ (see Section~\ref{subsec:AchThm2}). 

\begin{remark}
\emph{Interestingly, comparing the results of~\cite[Theorem~2]{Kadhe2017} and Theorem~\ref{thm:PIRPCSI-II}, one can see that when the user knows $M-1$ messages (different from the demand) separately, achieving $(W,S)$-privacy is as costly as that when the user's side information is only \emph{one} linear combination of $M$ messages including the demand.}
\end{remark}

\begin{remark}
\emph{As shown in~\cite[Theorem~2]{HKS2018}, when $W$-privacy is required, the capacity of single-server single-message PIR with a coded side information to which the demand message contributes (known as the PIR-CSI--II problem in~\cite{HKS2018}) is equal to $1$ for $M=2$ and $M=K$, and is equal to $\frac{1}{2}$ for all ${3\leq M\leq K-1}$. The result of Theorem~\ref{thm:PIRPCSI-II} matches this result for the cases of $M=K$ and $M=K-1$, and thereby, $(W,S)$-privacy and $W$-privacy are attainable at the same cost. For other cases of $M$, as expected, achieving $(W,S)$-privacy is more costly than achieving $W$-privacy.}
\end{remark}

\section{The PIR-PCSI--I Problem}\label{sec:PIRPCSI-I}

\subsection{Converse for Theorem~\ref{thm:PIRPCSI-I}}\label{subsec:ConvThm1}
Obviously, the capacity of PIR-PCSI--I is upper bounded by the capacity of PIR with uncoded side information where $(W,S)$-privacy is required, which was shown to be ${(K-M)^{-1}}$ in~\cite{Kadhe2017} using an index-coding argument, where $M$ uncoded messages are available at the user as side information. This proves the converse for Theorem~\ref{thm:PIRPCSI-I}. We present an alternative information-theoretic proof here. 

The following result gives a necessary condition for $(W,S)$-privacy and recoverability. 

\begin{lemma}\label{prop:1}
For any $\theta\in \{0,1\}$, $W\in \mathcal{K}$, and $S\in \mathcal{S}$ where $I^{[W,S]}=\theta$, and $C\in \mathcal{C}$, and any ${W^{*}\in \mathcal{K}}$ and $S^{*}\in \mathcal{S}$ where $I^{[W^{*},S^{*}]}=\theta$, there must exist ${C^{*}\in \mathcal{C}}$ such that \[H(X_{W^{*}}| A^{[W,S,C]}, Q^{[W,S,C]}, I^{[W,S]}, Y^{[S^{*},C^{*}]}) = 0.\] 
\end{lemma}

\begin{proof}
The proof is straightforward by the way of contradiction, and hence omitted.
\end{proof}

\begin{lemma}\label{lem:Conv1}
For any $0\leq M\leq K-1$, the capacity of PIR-PCSI--I is upper bounded by ${(K-M)^{-1}}$.
\end{lemma}

\begin{proof}
Fix $W$, $S$, and $C$ (and accordingly, $Y\triangleq Y^{[S,C]}$) such that $I^{[W,S]}=0$, and let $Q\triangleq Q^{[W,S,C]}$ and $A\triangleq A^{[W,S,C]}$ be the user's query and the server's answer, respectively, for an arbitrary PIR-PCSI-I protocol. We need to show that $H(A^{[\boldsymbol{W},\boldsymbol{S},\boldsymbol{C}]})=H(A)\geq (K-M)L$. Similar to the proof of \cite[Theorem~1]{HKS2018},  it can be shown that 
\begin{equation}\label{eq:line1}
H(A)\geq H(X_W)+H(A|Q,Y,X_W).	
\end{equation} If $W\cup S = \mathcal{K}$ (i.e., $M=K-1$), then we have $H(A)\geq H(X_W)=L$, as was to be shown. If $W\cup S \neq \mathcal{K}$, for any ${j\in \mathcal{K}\setminus (W\cup S)}$ there exists $C_{j}\in \mathcal{C}$ (and accordingly, $Y_j\triangleq Y^{[S,C_j]}$) such that $H(X_{j}|A,Q,Y_{j})=0$ (by Lemma~\ref{prop:1}). Let $I$ be a maximal subset of ${\mathcal{K}\setminus (W\cup S)}$ such that $Y$ and $Y_I\triangleq \{Y_{j}\}_{j\in I}$ are linearly independent. (Note that ${|I|\leq |S|-1=M-1}$.) Let $X_I\triangleq \{X_j\}_{j\in I}$. Then, we have
\begin{align}
H(A|Q,Y,X_W) &\geq H(A|Q,Y,X_W,Y_I)\nonumber \\
& = H(A|Q,Y,X_W,Y_I) \nonumber \\ 
& \quad + H(X_I|A,Q,Y,X_W,Y_I)\label{eq:line2}\\
& = H(X_I|Q,Y,X_W,Y_I)\nonumber\\
& \quad +  H(A|Q,Y,X_W,Y_I,X_I)\nonumber\\
& = H(X_I) +H(A|Q,Y,X_W,Y_I,X_I)\label{eq:line3}	
\end{align} where~\eqref{eq:line2} holds because $H(X_{j}|A,Q,Y_{j})=0$ for all $j\in I$ (by assumption); and~\eqref{eq:line3} holds since $X_I$ is independent of $(Q,Y,X_W,Y_I)$ (noting that $I$ and $W\cup S$ are disjoint). Note also that, by the maximality of $I$, for any $j\in J\triangleq {\mathcal{K}\setminus (W\cup S\cup I)}$, there exists $C_j\in C$ (and accordingly, $Y_j\triangleq Y^{[S,C_j]}$, which is linearly dependent on $\{Y,Y_I\}$) such that $H(X_j|A,Q,Y_j) = 0$, and subsequently, $H(X_j|A,Q,Y_I) = 0$. (Note that $|J|={K-M-1-|I|}$.) Thus, we can write
\begin{align}
& H(A|Q,Y,X_W,Y_I,X_I)\nonumber\\
&	\quad = H(A|Q,Y,X_W,Y_I,X_I)\nonumber\\
& \quad\quad + H(X_{J}|A,Q,Y,X_W,Y_{I},X_I)\label{eq:line4}\\
& \quad = H(X_J|Q,Y,X_W,Y_I,X_I)\nonumber\\
& \quad \quad + H(A|Q,Y,X_W,Y_I,X_I,X_J)\nonumber\\
& \quad \geq H(X_{J})\label{eq:line5}
\end{align} where~\eqref{eq:line4} holds since $H(X_{j}|A,Q,Y_{I})=0$ for all $j\in J$ (by assumption); and~\eqref{eq:line5} holds because $X_{J}$ and $(Q,Y,X_W,Y_I,X_I)$ are independent (noting that $J$ and ${W\cup S\cup I}$ are disjoint). Putting~\eqref{eq:line1}, \eqref{eq:line2}, \eqref{eq:line3}, and \eqref{eq:line5} together, it follows that $H(A)\geq H(X_W)+H(X_I)+H(X_{J}) = (K-M)L$, as was to be shown.
\end{proof}

\subsection{Achievability for Theorem \ref{thm:PIRPCSI-I}}\label{subsec:AchThm1}
In this section, we propose a PIR-PCSI--I protocol for arbitrary $K$ and $M$ that achieves the rate $(K-M)^{-1}$. Throughout, we assume that $q$ is sufficiently large, particularly $q\geq K$. For arbitrary $q<K$, the achievability of the rate $(K-M)^{-1}$, which is not necessarily feasible, is conditional on the existence of a $(K,K-M)$ maximum-distance-seperable (MDS) code over $\mathbb{F}_q$ that includes a codeword with support $S\cup W$ such that the $i$th codeword symbol is $c_i$ for $i\in S$, and is non-zero for $i=W$.

Assume that $q\geq K$, and let $\omega_1,\dots,\omega_K$ be $K$ distinct elements from $\mathbb{F}_q$. 

\textbf{Specialized GRS Code Protocol:} This protocol consists of four steps as follows: 

\emph{Step 1:} The user first constructs a polynomial  ${p(x) = \sum_{i=0}^{K-M-1} p_i x^i \triangleq \prod_{i\not\in S\cup W} (x-\omega_i)}$, and then constructs $K-M$ sequences $Q_1,\dots,Q_{K-M}$, each of length $K$, such that $Q_i=\{v_1\omega_1^{i-1},\dots,v_K\omega_K^{i-1}\}$ for $i\in [K-M]$, where $v_i=\frac{c_i}{p(\omega_i)}$ for $i\in S$, and $v_i$ is a randomly chosen element from $\mathbb{F}_q^{\times}$ for $i\not\in S$. 

For any ${i\in [K-M]}$, the $j$th element, for any ${j\in \mathcal{K}}$, in the sequence $Q_i$ can be thought of as the entry $(i,j)$ of a $(K-M)\times K$ matrix $G\triangleq {[g_1^{\mathsf{T}},\dots,g_{K-M}^{\mathsf{T}}]}^{\mathsf{T}}$, which is the generator matrix of a $(K,K-M)$ GRS code with distinct parameters ${\omega_1,\dots,\omega_{K}}$ and non-zero multipliers $v_1,\dots,v_K$~\cite{Roth:06}. The construction above ensures that such a GRS code has a specific codeword with support $S\cup W$, namely $\sum_{i=1}^{K-M} p_{K-M-i} g_i$, where the $i$th codeword symbol is $c_i$ for $i\in S$, and is non-zero for $i=W$. 

\emph{Step 2:} The user reorders $Q_1,\dots,Q_{K-M}$ by a randomly chosen permutation ${\sigma:[K-M]\rightarrow [K-M]}$, and sends the query $Q^{[W,S,C]} = \{Q_{\sigma^{-1}(1)},\dots,Q_{\sigma^{-1}(K-M)}\}$ to the server.

\emph{Step 3:} By using $Q_i$, the server computes $A_{i} = \sum_{j=1}^{K}  v_j\omega_j^{i-1} X_{j}$ for all $i\in [K-M]$ where $Q_{i} = \{v_1\omega_1^{i-1} ,\dots,v_K\omega_K^{i-1} \}$, and it sends the answer $A^{[W,S,C]}=\{A_{\sigma^{-1}(1)},\dots,A_{\sigma^{-1}(K-M)}\}$ to the user. 

Note that $A_i$'s are the parity check equations of a $(K,M)$ GRS code which is the dual code of the GRS code generated by the matrix $G$ defined earlier. 

\textbf{\it Step 4:} Upon receiving the answer, the user retrieves $X_W$ by subtracting off the contribution of the side information $Y^{[S,C]}$ from $\sum_{i=1}^{K-M} p_{K-M-i} A_{\sigma(i)} = c_W X_W+\sum_{i\in S} c_{i}X_{i}$.

\begin{lemma}\label{lem:Ach1}
The Specialized GRS Code protocol is a PIR-PCSI--I protocol, and achieves the rate $(K-M)^{-1}$. 
\end{lemma}

\begin{proof}
Since the matrix $G$, defined in Step~1 of the protocol, generates a $(K,K-M)$ GRS code which is an MDS code, then the rows of $G$ are linearly independent, and accordingly, $A_1,\dots,A_{K-M}$ are linearly independent combinations of $X_1,\dots,X_K$, which are themselves independently and uniformly distributed over $\mathbb{F}_{q^{m}}$. Thus, $A_1,\dots,A_{K-M}$ are independently and uniformly distributed over $\mathbb{F}_{q^{m}}$. Since $H(X_1)=\dots=H(X_K) = L$, then $H(A_1)=\dots=H(A_{K-M})=L$, and $H(A^{[W,S,C]}) = H(A_1,\dots,A_{K-M})=\sum_{i=1}^{K-M} H(A_i)=(K-M)L$ for any $S\in \mathcal{S}$, any $W\not\in S$, and any $C\in \mathcal{C}$. Since the joint distribution of $\boldsymbol{W}$ and $\boldsymbol{S}$ is uniform and $\boldsymbol{C}$ is uniformly distributed, then $H(A^{[\boldsymbol{W},\boldsymbol{S},\boldsymbol{C}]})=H(A^{[W,S,C]})$. Thus, the Specialized GRS Code protocol has the rate $L/H(A^{[\boldsymbol{W},\boldsymbol{S},\boldsymbol{C}]}) = L/H(A^{[W,S,C]}) = (K-M)^{-1}$. 

Next, we prove that the Specialized GRS Code protocol is a PIR-PCSI--I protocol. It should be obvious from the construction that the recoverability condition is satisfied. The $(W,S)$-privacy condition is also satisfied because the $(K,K-M)$ GRS code, generated by the matrix $G$, is an MDS code, and thereby, the minimum (Hamming) weight of a codeword is $K-(K-M)+1 = M+1$, and there are the same number of minimum-weight codewords for any support of size ${M+1}$~\cite{Roth:06}. Thus, for any $S\in \mathcal{S}$ and any $W\not\in S$, the dual code, whose parity check matrix is given by $G$, contains the same number of parity check equations (with support $S\cup W$) from each of which, given $Y^{[S,C]}$ for some $C\in\mathcal{C}$, $X_W$ can be recovered.   
\end{proof}

\section{The PIR-PCSI--II Problem}\label{sec:PIRPCSI-II}

\subsection{Converse for Theorem~\ref{thm:PIRPCSI-II}}\label{subsec:ConvThm2}
In this section, we give an information-theoretic proof of converse for Theorem~\ref{thm:PIRPCSI-II}.

\begin{lemma}\label{lem:Conv2}
For any $2\leq M\leq K$, the scalar-linear capacity of PIR-PCSI--II is upper bounded by ${(K-M+1)^{-1}}$.	
\end{lemma}

\begin{proof}
Fix $W$, $S$, and $C$ (and $Y\triangleq Y^{[S,C]}$) such that $I^{[W,S]}=1$. Let $Q\triangleq Q^{[W,S,C]}$ and $A\triangleq A^{[W,S,C]}$ be the query and the answer of an arbitrary scalar-linear PIR-PCSI--II protocol. We need to show that $H(A)\geq {(K-M+1)L}$. Let $I$ be the set of all $j\in \mathcal{K}$ such that $H(X_j|A,Q)=0$, i.e., $X_j$ is recoverable from $A$ (and $Q$) directly. Let $X_{I}\triangleq \{X_j\}_{j\in I}$. There are two cases: (i) $I\neq \emptyset$, and (ii) $I=\emptyset$.   

\emph{Case (i):} Since $X_I$ and $Q$ are independent and $H(X_I|A,Q)=0$ (by assumption), then 
\begin{align}
H(A) & \geq H(A|Q) +H(X_I|A,Q)\nonumber\\
& = H(X_I|Q)+H(A|Q,X_I)\nonumber\\
& = H(X_I)+H(A|Q,X_I).\label{eq:line6}
\end{align} If ${|I|\geq K-M+1}$, then $H(X_I)\geq (K-M+1)L$, and subsequently, $H(A)\geq (K-M+1)L$, as was to be shown. If $|I|\leq K-M$, $H(A|Q,X_I)$ can be further lower bounded as follows. Let $n\triangleq |I|$. Assume, w.l.o.g., that $I=[n]$. Let $J\triangleq [K-M-n+1]$, and $S_j\triangleq {\{n+1,n+j+1,\dots,n+j+M-1\}}$ for $j\in J$. (Note that $|J|=K-M-n+1$.) By Lemma~\ref{prop:1}, for any $j\in J$, there exists $C_j\in \mathcal{C}$ (and accordingly, $Y_j\triangleq Y^{[S_j,C_j]}$) such that $H(X_{n+1}|A,Q,Y_j)=0$. Let $Z_j\triangleq Y_j - c_jX_{n+1}$ where $c_j$ is the coefficient of $X_{n+1}$ in $Y_j$. By the scalar-linearity of $A$, it is easy to see that either $H(Z_j|A,Q)=0$ or ${H(Z_j+c^{*}_jX_{n+1}|A,Q)=0}$ for some $c^{*}_j\in \mathbb{F}^{\times}_q\setminus \{c_j\}$. (Otherwise, the server learns that the user's demand index and side information index set cannot be $n+1$ and $S_j$, respectively. This obviously violates the $(W,S)$-privacy condition.) Thus, $H(Z_j|A,Q,X_{n+1})=0$. Let $Z_{J}\triangleq \{Z_j\}_{j\in J}$. Then, we have
\begin{align}
H(A|Q,X_I) & \geq H(A|Q,X_I,X_{n+1})\nonumber \\
& = H(A|Q,X_I,X_{n+1})\nonumber\\
& \quad + H(Z_J|A,Q,X_I,X_{n+1})\label{eq:line7}\\
& = H(Z_J|Q,X_I,X_{n+1}) \nonumber \\
&\quad + H(A|Q,X_I,X_{n+1},Z_J)\nonumber\\
& \geq H(Z_J) \label{eq:line8}
\end{align} where~\eqref{eq:line7} holds since $H(Z_j|A,Q,X_{n+1})=0$ for all $j\in J$ (by assumption); and~\eqref{eq:line8} follows because $Z_J$ is independent of $(Q,X_I,X_{n+1})$, noting that $Z_J$, $X_I$, and $X_{n+1}$ are linearly independent (by construction). By the linear independence of $Z_j$'s for all $j\in J$, it follows that $H(Z_J) = {(K-M-n+1)L}$. By~\eqref{eq:line6} and~\eqref{eq:line8}, we get $H(A)\geq {nL}+{(K-M-n+1)L} ={ (K-M+1)L}$.  

\emph{Case (ii):} Assume, w.l.o.g., that ${W=1}$ and $S=[M]$. Let $J\triangleq [K-M]$, and $S_j\triangleq {\{1,j+2,\dots,j+M-2\}}$ for $j\in J$. (Note that $|J|=K-M$.) Similarly as in the case (i), define $Y_j$ (and accordingly $Z_j$) for all $j\in J$, where $X_{n+1}$ is replaced by $X_1$. By using a similar argument as before, it can be shown that $H(Z_j|A,Q,X_1)=0$ for all $j\in J$. Let $Z_{J}\triangleq \{Z_j\}_{j\in J}$. Then, we can write
\begin{align}
H(A) &\geq H(A|Q,Y)\nonumber \\
& = H(A|Q,Y)+H(X_1|A,Q,Y)\label{eq:line9}\\
& = H(X_1|Q,Y)+H(A|Q,Y,X_1)\nonumber \\
& = H(X_1) + H(A|Q,Y,X_1)\nonumber \\
& \quad +H(Z_J|A,Q,Y,X_1)\label{eq:line10}\\
& = H(X_1)+H(Z_J|Q,Y,X_1)\nonumber\\
& \quad + H(A|Q,Y,X_1,Z_J)\nonumber\\
& \geq H(X_1)+H(Z_J)\label{eq:line11}
\end{align} where~\eqref{eq:line9} follows since $H(X_1|A,Q,Y)=0$ (by the recoverability condition);~\eqref{eq:line10} holds because ${H(Z_j|A,Q,X_1)=0}$, and subsequently, $H(Z_j|A,Q,Y,X_1)=0$, for all $j\in J$; and~\eqref{eq:line11} follows because $Z_J$ is independent of $(Q,Y,X_1)$ (due to the linear independence of $Z_J$, $Y$, and $X_1$). Since $|J|=K-M$, we have $H(Z_J) = (K-M)L$ (noting that $Z_j$'s are linearly independent), and thereby, $H(A)\geq L+(K-M)L=(K-M+1)L$.
\end{proof}

\subsection{Achievability for Theorem \ref{thm:PIRPCSI-II}}\label{subsec:AchThm2}
In this section, we propose a PIR-PCSI--II protocol, which is a slightly modified version of the Specialized GRS Code protocol, that achieves the rate $(K-M+1)^{-1}$ for arbitrary $K$ and $M$. 

\textbf{Modified Specialized GRS Code Protocol:} This protocol consists of four steps, where the steps 2-4 are the same as those in the Specialized GRS Code protocol (Section~\ref{subsec:AchThm1}), except that $M$ is replaced with $M-1$ everywhere. The step~1 of the proposed protocol is as follows: 

\emph{Step 1:} The user first constructs a polynomial  ${p(x) = \sum_{i=0}^{K-M} p_i x^i \triangleq \prod_{i\not\in S} (x-\omega_i)}$, and then constructs $K-M+1$ sequences $Q_1,\dots,Q_{K-M+1}$, each of length $K$, such that $Q_i=\{v_1\omega_1^{i-1},\dots,v_K\omega_K^{i-1}\}$ for $i\in [K-M]$, where $v_i=\frac{c_i}{p(\omega_i)}$ for $i\in S\setminus W$; $v_W=\frac{c}{p(\omega_W)}$ where $c$ is chosen uniformly at random from $\mathbb{F}^{\times}_q\setminus \{c_W\}$; and $v_i$ is a randomly chosen element from $\mathbb{F}_q^{\times}$ for $i\not\in S$. 

\begin{lemma}\label{lem:Ach2}
The Modified Specialized GRS Code protocol is a PIR-PCSI--II protocol, and achieves the rate $(K-M+1)^{-1}$. 
\end{lemma}

\begin{proof}
The proof, omitted to avoid repetition, follows from the same lines as in the proof of Lemma~\ref{lem:Ach1} where $M$ is replaced by $M-1$, and $W\not\in S$ is replaced by $W\in S$. 
\end{proof}

\bibliographystyle{IEEEtran}
\bibliography{PIR_salim,pir_bib,coding1,coding2}

\end{document}